\documentclass{amsart}

\usepackage{amsthm}
\usepackage{amsmath} \usepackage{amsfonts}
\usepackage{amssymb} \usepackage{latexsym} \usepackage{enumerate}
\usepackage{mathtools}

\usepackage{mathrsfs}
\usepackage[numbers]{natbib}

\numberwithin{equation}{section}
\newtheorem{thm}{Theorem}[section]

\newtheorem{lem}[thm]{Lemma}
\newtheorem{remark}[thm]{Remark}

\newtheorem{prop}[thm]{Proposition}
\newtheorem{assumption}[thm]{Assumption}

\begin{document}
\title{Duality Theory for Exponential Utility-Based Hedging in the Almgren-Chriss Modell}
  \thanks{Supported in part by the GIF Grant 1489-304.6/2019 and the ISF grant 230/21}

   \author{Yan Dolinsky \address{
 Department of Statistics, Hebrew University of Jerusalem.\\
 e.mail: yan.dolinsky@mail.huji.ac.il}}
  \date{\today}
\maketitle
\begin{abstract}
In this paper, we obtain a
duality result for the
exponential utility maximization problem where trading is subject to quadratic transaction costs
and the investor is required to liquidate her position at the maturity date.
As an application of the duality, we treat
utility-based hedging in the
Bachelier model. For European contingent claims with a quadratic payoff, we compute explicitly
the optimal trading strategy.
\end{abstract}
\begin{description}
\item[Mathematical Subject Classification (2020)] 91B16, 91G10, 60H30
\item[Keywords] exponential utility, duality, linear price impact, Bachelier model
\end{description}

 \maketitle \markboth{}{}
\renewcommand{\theequation}{\arabic{section}.\arabic{equation}}
\pagenumbering{arabic}

\section{Introduction}\label{sec:1}\setcounter{equation}{0}
In financial markets, trading moves prices against the trader:
buying faster increases execution prices, and selling faster decreases them. This aspect
of liquidity is known as market depth
\cite{B} or price-impact.
In this paper we consider
the problem of optimal liquidation
for the exponential utility function in
the Almgren--Chriss model \cite{AlmgrenChriss:01}
with linear temporary impact for the underlying asset.

This problem
goes back to  \cite{SST:09} where the authors considered a market model given by a L\'{e}vy process.
The authors proved that the optimal trading strategy is deterministic and hence reduced the primal problem to
a deterministic variational problem that can be solved explicitly.
A similar phenomenon occurs in \cite{BV:2019} where the authors consider an optimal liquidation problem
with transient price impact in the Bachelier model. Namely, the fact that the utility function is
exponential and the risky asset has independent increments allows us to reduce the primal hedging problem to a deterministic control problem.

For the case where the market model is not given by a process with independent
increments, the exponential utility maximization problem
in the Almgren--Chriss model is much more complicated and
typically does not have an explicit solution.
In \cite{GS:11} the authors found a closed-form solution
for the optimal trade execution strategy in the Almgren–Chriss
framework assuming that the risky asset is given by the
Black--Scholes model.
However, the risk criterion the authors used was given
by the expected value of the terminal wealth and
hence, it is analytically simpler than the exponential utility maximization problem.
In general, although the current paper is focused on the
Almgren--Chriss model, there are other models for optimal liquidation problems.
For instance, one common approach is via limit order books (see \cite{BL:14,FSU:19}
and the references therein).

Our first result is Theorem \ref{thm2.1}
which provides a dual representation
for the optimal portfolio
and the corresponding value of the
exponential utility maximization problem.
Our duality result is obtained under quite general
assumptions on the market model.
As usual, for the case of exponential utility, by applying a change of measure
one can reduce the problem of utility-based hedging of
a European contingent claim to the standard utility maximization problem.
This brings us to our second result.

Our second result (Theorem \ref{thm3.1})
deals with explicit computations for the
case where the risky asset is given by a linear Brownian motion,
i.e. the Bachelier model.
We consider a European contingent claim with the payoff given
by $\kappa S^2_T$ where $\kappa>0$ is a constant
and $S_T$ is the
stock price at the maturity date.
We apply the Girsanov theorem and the It\^{o} Isometry in order to
derive a particularly convenient representation of the dual target functional
which leads to deterministic variational problems. These problems can be solved
explicitly and allow us both to construct the solution to the dual problem and to compute
the primal optimal strategy. We show that the
optimal strategy is given by a feedback form which we compute explicitly.
For the case $\kappa=0$, i.e. there is no option,
Theorem \ref{thm3.1} recovers the optimal portfolio found in
\cite{SST:09} for the Bachelier model.

The problem of utility-based hedging for the Almgren--Chriss model
in the Bachelier setup was studied recently in \cite{EN:2022}.
 The authors apply the Hamilton--Jacobi-Bellman (HJB)
methodology and obtain a representation of the value function
and the optimal strategy. Still, the authors do
not require the liquidation of the portfolio at the maturity date.
Moreover, the authors assume that the payoff function is globally Lipschitz.

A natural question that for now remains open is whether Theorem \ref{thm2.1}
can be applied beyond
the Bachelier model. In particular, it is not clear
whether by applying this duality result one can recover
the optimal portfolio from \cite{SST:09}
for a general L\'{e}vy process (beyond Brownian motion).

The rest of the paper is organized as follows. In Section \ref{sec:2}
we introduce the model and formulate a general duality result (Theorem \ref{thm2.1}).
In Section \ref{sec:3}
we consider the Bachelier model and we explicitly solve the problem of utility-based hedging for European contingent claims with a quadratic payoff (Theorem \ref{thm3.1}).
In Section \ref{sec:4} we derive an
auxiliary result from the field of deterministic variational analysis.  

\section{Preliminaries and the Duality Result}\label{sec:2}
Let
$(\Omega,\mathcal F,(\mathcal F_t)_{t\in [0,T]},\mathbb P)$ be a filtered probability space
equipped with the completed and right--continuous filtration $(\mathcal F_t)_{t\in [0,T]}$ and without loss of generality we assume
that
$\mathcal F=\mathcal F_T$. We do not make any assumptions on $\mathcal F_0$.
Consider a simple financial market
with a riskless savings account bearing zero interest (for simplicity) and with a
RCLL (right continuous with left
limits) risky asset $S=(S_t)_{t\in [0,T]}$
which is adapted to the filtration
$(\mathcal F_t)_{t\in [0,T]}$. We assume the following growth condition.
\begin{assumption}\label{asm3.1}
There is $a>0$ such that
$\mathbb E_{\mathbb P}\left[\exp\left(a\sup_{0 \leq t\leq T} S^2_t\right)\right]<\infty.$
\end{assumption}

Following \cite{AlmgrenChriss:01}, we model the investor’s market impact in a
temporary linear form and thus, when at time $t$ the investor turns over her position $\Phi_t$ at the rate $\phi_t=\dot{\Phi}_t$
the execution price is $S_t+\frac{\Lambda}{2}\phi_t$ for some constant $\Lambda>0$. In our setup the investor has to liquidate
his position, namely $\Phi_T=\Phi_0+\int_{0}^T\phi_t dt=0$.

As a result, the profits and
losses from trading are given by
\begin{equation}\label{2.1}
V^{\Phi_0,\phi}_T:=-\Phi_0 S_0-\int_{0}^T \phi_t S_t dt-\frac{\Lambda}{2}\int_{0}^T \phi^2_t dt
\end{equation}
where $\Phi_0$ is the initial number (deterministic) of shares.
Observe that all the above integrals are defined pathwise.
In particular we do not assume that $S$ is a semi--martingale.

\begin{remark}
Let us explain in more detail formula (\ref{2.1}). At time $0$ the investor has $\Phi_0$ stocks and
the sum $-\Phi_0 S_0$ on her savings account. At time $t\in [0,T)$
the investor buys $\phi_t dt$, an infinitesimal number of stocks or more intuitively
sell $-\phi_t dt$ number of shares and so the (infinitesimal) change in the savings account
is given by $-\phi_t \left(S_t+\frac{\Lambda}{2}\phi_t\right)dt$. Since we liquidate
the portfolio at the maturity date, the terminal portfolio value is equal to
the terminal amount on the savings account and given by
$-\Phi_0 S_0- \int_{0}^T \phi_t \left(S_t+\frac{\Lambda}{2}\phi_t\right)dt$. We arrive at the right-hand side of (\ref{2.1}).
For the case where $S$ is a semi--martingale, by applying the integration by parts formula we get that the right-hand side of (\ref{2.1}) is equal to
$\int_{0}^T \Phi_t dS_t-\frac{\Lambda}{2} \int_{0}^T \phi^2_t dt$.
\end{remark}

For a given $\Phi_0$, the natural class of admissible strategies is
$$\mathcal A_{\Phi_0}:=\left\{\phi \ : \ \phi \ \mbox{is} \ (\mathcal F_t)_{t\in [0,T]}-\mbox{optional} \ \mbox{with} \ \int_{0}^T \phi^2_t  dt<\infty \  \mbox{and} \ \Phi_0+\int_{0}^T \phi_t dt=0  \right\}.$$
As usual, all the equalities and the inequalities are understood in the almost surely sense.

The investor’s preferences are described by an exponential utility function
$u(x)=-\exp(-\alpha x)$, $x\in\mathbb R$,
with constant absolute risk aversion parameter $\alpha>0$, and for a given $\Phi_0$ her goal is to
\begin{equation}\label{2.2}
\mbox{Maximize} \  \mathbb{E}_{\mathbb P}\left[-\exp\left(-\alpha V^{\Phi_0,\phi}_T\right) \right] \ \mbox{over} \
\phi\in\mathcal A_{\Phi_0}.
\end{equation}

Next, we introduce some notations.
Denote by $\mathcal Q$ the set of all equivalent
probability measures $\mathbb Q\sim\mathbb P$ with finite entropy
$\mathbb E_{\mathbb Q}\left[\log\left(\frac{d\mathbb Q}{d\mathbb P}\right)\right]<\infty$
relative to $\mathbb P$.
For any $\mathbb Q\in\mathcal Q$ let $\mathcal M^{\mathbb Q}_{[0,T]}$ be the set of all square integrable $\mathbb Q$--martingales
$M=(M_t)_{0\leq t\leq T}$.
 Moreover, let $\mathcal M^{\mathbb Q}_{[0,T)}$
be the set of all $\mathbb Q$--martingales $M=(M_t)_{0\leq t<T}$ which are defined on the half-open interval $[0,T)$
and satisfy
$||M||_{L^2(dt\otimes\mathbb Q)}:=\mathbb E_{\mathbb Q}\left[\int_{0}^T M^2_t dt\right]<\infty$.

We arrive at the duality result.
\begin{thm}\label{thm2.1}
Let Assumption \ref{asm3.1} be in force.
Then for any $\Phi_0\in\mathbb R$,
\begin{align}\label{3.1}
&\max_{\phi\in\mathcal A_{\Phi_0}}\left\{-\frac{1}{\alpha}\log\mathbb{E}_{\mathbb P}\left[\exp\left(-\alpha V^{\Phi_0,\phi}_T\right) \right]\right\}\nonumber\\
&=\inf_{\mathbb Q\in\mathcal Q}\inf_{M\in\mathcal M^{\mathbb Q}_{[0,T)}}\left\{\mathbb E_{\mathbb Q}\left[\frac{1}{\alpha}\log\left(\frac{d\mathbb Q}{d\mathbb P}\right)+\Phi_0(M_0-S_0)+\frac{1}{2\Lambda}\int_{0}^T |M_t-S_t|^2 dt\right]\right\}.
\end{align}
Furthermore, there is a unique minimizer
$\left(\hat{\mathbb Q},\hat M\in \mathcal M^{\hat{\mathbb Q}}_{[0,T)}\right)$ for the dual problem (right-hand side of (\ref{3.1}))
 and the process
given by
\begin{equation}\label{3.2}
\hat\phi_t:=\frac{\hat M_t-S_t}{\Lambda}, \ \ t\in [0,T)
\end{equation}
is the unique optimal portfolio ($dt\otimes\mathbb P$ a.s.) for the primal problem (\ref{2.2}).
\end{thm}
\begin{remark}
Let us notice that it is sufficient to define the optimal portfolio on the
half-open interval $[0,T)$. We can just set $\phi_T:=0$.
\end{remark}
\begin{remark}
Theorem \ref{thm2.1} can be viewed as an extension
of Proposition A.2 in \cite{BDR:22} for the
case where the investor liquidates
her portfolio at the maturity date.
The liquidation requirement adds additional difficulty
to the dual representation.
In particular the maximization in the dual representation
is over all equivalent probability measures and the corresponding
martingales, in contrast
to Proposition A.2 in \cite{BDR:22} where the dual objects are
just equivalent probability measures.

The duality result in \cite{BDR:22} was applied for solving the problem
of exponential utility maximization in the Bachelier setting for the case
where the investor can
peek some time units into the future (frontrunner). Theorem \ref{thm2.1}
allows us to solve the same problem with the additional requirement
that the portfolio has to be liquidated at the maturity date.
Since the corresponding computations are not straightforward we leave this problem for future research.

In the proof of the duality we assume that $\Lambda>0$. However if we formally take $\Lambda=0$ in the right-hand side of
(\ref{3.1}) and use the convention $\frac{0}{0}:=0$ we get the relation
$\max_{\phi}\left\{-\log\mathbb{E}_{\mathbb P}\left[\exp\left(-\alpha V^{\Phi_0,\phi}_T\right) \right]\right\}
=\inf_{\mathbb Q}\mathbb E_{\mathbb Q}\left[\log\left(\frac{d\mathbb Q}{d\mathbb P}\right)\right]$
where the infimum is taken over all martingale measures. This is (roughly speaking) the classical duality result
for exponential hedging in the frictionless setup (see \cite{D:02, F:00}).
Of course, in the frictionless setup there is no meaning to the initial number of shares
$\Phi_0$ and there is no real restriction in the requirement $\Phi_T=0$.
\end{remark}
We will prove Theorem \ref{thm2.1} at the end of this section, after suitable preparations.
We start with proving the superhedging theorem.
\begin{lem}\label{lem3.1}
Let $X$ be a random variable. Assume that there exists $\alpha>0$ for which
\begin{equation}\label{3.1+}
\mathbb E_{\mathbb P}\left[\exp\left(\alpha \max(-X,0)\right)\right]<\infty.
\end{equation}
There exists $\phi\in\mathcal A_{\Phi_0}$
such that $V^{\Phi_0,\phi}_T\geq X$ if and only if
\begin{equation}\label{3.2++}
\sup_{\mathbb Q\in\mathcal Q}\sup_{M\in\mathcal M^{\mathbb Q}_{[0,T]}}\mathbb E_{\mathbb Q}
\left[X-\Phi_0(M_0-S_0)-\frac{1}{2\Lambda}\int_{0}^T |M_t-S_t|^2 dt\right]\leq 0.
\end{equation}

\end{lem}
\begin{proof}
We start with the "only if" part of the claim.
Let $\phi\in\mathcal A_{\Phi_0}$ such that
$V^{\Phi_0,\phi}_T\geq X$.
Choose $\mathbb Q\in \mathcal Q$ and
$M\in\mathcal M^{\mathbb Q}_{[0,T]}$.
From (\ref{2.1}) and the
Cauchy–Schwarz inequality
it follows
that
$$
\sqrt{\int_{0}^T S^2_t dt}  \sqrt{\int_{0}^T \phi^2_t dt}-\frac{\Lambda}{2}\int_{0}^T \phi^2_t dt\geq
X+\Phi_0S_0$$
and so, by exploiting the behaviour of the (random) parabola $x\rightarrow x \sqrt{\int_{0}^T S^2_t dt} -\frac{\Lambda}{2} x^2$, we get that
$\int_{0}^T \phi^2_t dt\leq c\left(1+\max(-X,0)+\int_{0}^T S^2_t dt\right)$ for some constant $c>0$.
This together with Assumption \ref{asm3.1}, (\ref{3.1+}) and the well known inequality
$e^x+y\log y\geq xy$, $x\in\mathbb R$, $y>0$
yields
$\mathbb E_{\mathbb Q}\left[\int_{0}^T \phi^2_t dt \right]<\infty.$
Hence,
$\mathbb E_{\mathbb Q}\left[\Phi_0M_0+\int_{0}^T M_t\phi_t dt\right]=0$, and so
from (\ref{2.1}) and the simple inequality
$x y-\frac{\Lambda}{2}x^2\leq \frac{y^2}{2\Lambda}$, $x,y\in\mathbb R$ we obtain
\begin{align*}
&\mathbb E_{\mathbb Q}[X]\leq\mathbb E_{\mathbb Q}\left[\Phi_0 (M_0-S_0)+\int_{0}^T \phi_t (M_t-S_t) dt- \frac{\Lambda}{2}\int_{0}^T \phi^2_t dt\right]\\
&\leq \mathbb E_{\mathbb Q}\left[\Phi_0 (M_0-S_0)+\frac{1}{2\Lambda}\int_{0}^T |M_t-S_t|^2 dt\right]
\end{align*}
and the result follows.

Next, we prove
the "if" part of the claim.
Assume by contradiction that
this part does not hold true. Namely,
 there exists $X$ which satisfies
(\ref{3.1+})--(\ref{3.2++}) and there is no $\phi\in\mathcal A_{\Phi_0}$
such that $V^{\Phi_0,\phi}_T\geq X$.

Following
the same arguments as in
the proof of Proposition 3.5 in \cite{GR:15}
it follows that the set $\Upsilon:=\left(\left\{V^{\Phi_0,\phi}_T:\phi\in\mathcal A_{\Phi_0}\right\}-L^0_{+}(\mathbb P)\right)\cap L^1(\mathbb P)$ is convex and closed in $L^1(\mathbb P)$. Observe that from Assumption \ref{asm3.1}
and (\ref{3.1+})--(\ref{3.2++}) (take $\mathbb Q=\mathbb P$ and $M\equiv 0$ in (\ref{3.2++})) it follows that
$X\in L^1(\mathbb P)$. Since there is no $\phi\in\mathcal A_{\Phi_0}$
such that $V^{\Phi_0,\phi}_T\geq X$ we get that $X\in L^1(\mathbb P)\setminus\Upsilon$.

Thus, by the Hahn-Banach Separation
Theorem we can find $Z\in L^{\infty}\setminus\{0\}$
such that
$\mathbb E_{\mathbb P}[Z X]>\sup_{\upsilon\in \Upsilon}\mathbb E_{\mathbb P}\left[Z \upsilon\right].$
Since $\left(V^{\Phi_0,\hat\phi}_T-L^0_{+}(\mathbb P)\right)\cap L^1(\mathbb P)\subset\Upsilon$ for $\hat\phi\equiv -\frac{\Phi_0}{T}$
we must have $Z\geq 0$.
Moreover, from (\ref{2.1}) we have $V^{\Phi_0,\phi}_T\leq -\Phi_0S_0+\frac{1}{2\Lambda}\int_{0}^T S^2_t dt$
for all $\phi$ and so,
from Assumption \ref{asm3.1}
it follows that there exists $\epsilon>0$ such that
$\mathbb E_{\mathbb P}[(\epsilon+Z) X]>\sup_{\upsilon\in \Upsilon}\mathbb E_{\mathbb P}\left[(\epsilon+Z) \upsilon\right].$
We conclude that for the probability measure $\mathbb Q$ given by
$\frac{d\mathbb Q}{d\mathbb P}:=\frac{\epsilon+ Z}{\epsilon +\mathbb  E_{\mathbb P}[Z]}$
we have $\mathbb Q\in \mathcal Q$ and
\begin{equation}\label{3.5}
\mathbb E_{\mathbb Q}[X]>\sup_{\phi\in\mathcal A_{\Phi_0}}\mathbb E_{\mathbb Q}\left[V^{\Phi_0,\phi}_T\right]
\end{equation}
where
we set $\mathbb E_{\mathbb Q}\left[V^{\Phi_0,\phi}_T\right]:=-\infty$ if
$V^{\Phi_0,\phi}_T\notin L^1(\mathbb Q)$.

Next, fix $n\in\mathbb N$ and introduce the set
$\mathcal B_n:=\left\{\phi\in L^2(dt\otimes\mathbb Q): \ ||\phi||_{L^2(dt\otimes\mathbb Q)}\leq n\right\}$.
We argue that for any $n\in\mathbb N$
\begin{align}\label{3.6}
&\sup_{\phi\in\mathcal A_{\Phi_0}}\mathbb E_{\mathbb Q}\left[V^{\Phi_0,\phi}_T\right]\nonumber\\
&\geq\sup_{\phi\in\mathcal A_{\Phi_0} \cap \mathcal B_n}\mathbb E_{\mathbb Q}\left[V^{\Phi_0,\phi}_T\right]\nonumber\\
&=\sup_{\phi\in  \mathcal B_n}\inf_{M\in\mathcal M^{\mathbb Q}_{[0,T]}}
\mathbb E_{\mathbb Q}\left[V^{\Phi_0,\phi}_T+M_T\left(\Phi_0+\int_{0}^T\phi_t dt\right)\right]\nonumber\\
&=\inf_{M\in\mathcal M^{\mathbb Q}_{[0,T]}}\sup_{\phi\in \mathcal B_n}
\mathbb E_{\mathbb Q}\left[V^{\Phi_0,\phi}_T+M_T\left(\Phi_0+\int_{0}^T\phi_t dt\right)\right].\nonumber\\
&
\end{align}
Indeed, the inequality is obvious. The first equality follows from the fact that if
$\phi\in \mathcal B_n\setminus \mathcal A_{\Phi_0}$ then
$\inf_{M\in\mathcal M^{\mathbb Q}_{[0,T]}}
\mathbb E_{\mathbb Q}\left[M_T\left(\Phi_0+\int_{0}^T\phi_t dt\right)\right]=-\infty$.
For the last equality in (\ref{3.6}) we apply a minimax theorem.
Consider the vector space $\mathcal M^{\mathbb Q}_{[0,T]}$
with the
$L^2(dt\otimes\mathbb Q)$ norm
and the set
$\mathcal B_n$
with the weak topology which corresponds to $L^2(dt\otimes\mathbb Q)$.
Then both of these sets are convex subsets of topological
vector spaces and the latter set is even compact.
Moreover,
$(\phi,M)\rightarrow \mathbb E_{\mathbb Q}\left[V^{\Phi_0,\phi}_T+M_T\left(\Phi_0+\int_{0}^T\phi_t dt\right)\right]$
is upper semi-continuous and concave in $\phi$ and
convex (even affine) in $M$. We can thus apply Theorem 4.2 in \cite{S:58} to obtain the second equality.

Next, choose $M\in\mathcal M^{\mathbb Q}_{[0,T]}$
and introduce the process $\phi$ (which depends on $n$ and $M$)
$$\phi_t:=\frac{n(M_t-S_t)}{\max\left(\Lambda n,||M-S||_{{L^2(dt\otimes\mathbb Q)}}\right)}, \ \ t\in [0,T].$$
Observe that $\phi\in \mathcal B_n$ and simple computations give
\begin{align*}
&\mathbb E_{\mathbb Q}\left[V^{\Phi_0,\phi}_T+M_T\left(\Phi_0+\int_{0}^T\phi_t dt\right)\right]\\
&=\mathbb E_{\mathbb Q}\left[\Phi_0 (M_0-S_0)+\int_{0}^T \phi_t (M_t-S_t) dt- \frac{\Lambda}{2}\int_{0}^T \phi^2_t dt\right]\\
&=\mathbb E_{\mathbb Q}\left[\Phi_0 (M_0-S_0)\right]+G_n\left( ||M-S||_{L^2(dt\otimes\mathbb Q)}\right)
\end{align*}
where
$$G_n(x) = \begin{cases}
  \frac{x^2}{2\Lambda}  & \mbox{if} \ \ x<\Lambda n\\
  n x-\frac{\Lambda}{2} n^2 & \ \mbox{otherwise}.
\end{cases}
$$
Since $n\in\mathbb N$ and $M\in\mathcal M^{\mathbb Q}_{[0,T]}$ were arbitrary, from (\ref{3.5})--(\ref{3.6}) we conclude that there exists a sequence of martingales
$M^n\in\mathcal M^{\mathbb Q}_{[0,T]}$, $n\in\mathbb N$ such that
\begin{equation}\label{3.6+}
\mathbb E_{\mathbb Q}[X]> \sup_{n\in\mathbb N} \left\{\mathbb E_{\mathbb Q}\left[\Phi_0 (M^n_0-S_0)\right]+G_n\left(||M^n-S||_{L^2(dt\otimes\mathbb Q)}\right)\right\}.
\end{equation}
From the fact that $\frac{d\mathbb Q}{d\mathbb P}$ is bounded we have
$\mathbb E_{\mathbb Q}[X]<\infty$, and so
$\sup_{n\in\mathbb N}||M^n-S||_{L^2(dt\otimes\mathbb Q)}<\infty$.
Thus, from (\ref{3.6+}) we obtain that
 for any $k>\frac{1}{\Lambda}\sup_{n\in\mathbb N}||M^n-S||_{L^2(dt\otimes\mathbb Q)}$
$$\mathbb E_{\mathbb Q}[X]>\mathbb E_{\mathbb Q}\left[\Phi_0 (M^k_0-S_0)+\frac{1}{2\Lambda}\int_{0}^T |M^k_t-S_t|^2 dt\right]$$
which is a contradiction to
(\ref{3.2++}).
This completes the proof.
\end{proof}
\begin{lem}\label{lem3.2}
There exists a unique minimizer
$\left(\hat{\mathbb Q},\hat M\in \mathcal M^{\hat{\mathbb Q}}_{[0,T)}\right)$ for the optimization problem
given by the right-hand side of (\ref{3.1}).
\end{lem}
\begin{proof}
Denote by $C$ the set of all pairs $(Z,Y)$ such that
$Z>0$ is a random variable which satisfies
 $\mathbb E_{\mathbb P}[Z]=1$,
 $\mathbb E_{\mathbb P}\left[ Z\log Z\right]<\infty$
 and
$Y=(Y_t)_{0\leq t< T}$ is a $\mathbb P$--martingale which satisfies
$\mathbb E_{\mathbb P}\left[\int_{0}^T \frac{Y^2_t}{\mathbb E_{\mathbb P}\left[Z|\mathcal F_t\right]}dt\right]<\infty$.
Let us notice that
the function $(z,y)\rightarrow y^2/z$ is convex on $\mathbb R_{++}\times\mathbb R$
and so $C$ is a convex set.
Define a map $\Psi:C\rightarrow \mathbb R$
by
$$\Psi(Z,Y):=\mathbb E_{\mathbb P}\left[\frac{1}{\alpha} Z\log Z+\Phi_0\left(Y_0-S_0 Z\right)+
\frac{Z}{2\Lambda}\int_{0}^T
\left(\frac{Y_t}
{\mathbb E_{\mathbb P}[Z|\mathcal F_t]}-S_t\right)^2 dt\right].$$

Observe that there is a bijection
$(Y,Z)\in C\leftrightarrow \left(\mathbb Q\in\mathcal Q,M\in\mathcal M^{\mathbb Q}_{[0,T)}\right)$
given by
$Z:=\frac{d\mathbb Q}{d\mathbb P}$ and
$M_t:=\frac{Y_t}{\mathbb E_{\mathbb P}[Z|\mathcal F_t]}$, $t\in [0,T)$
and for this bijection we have
$$\Psi(Z,Y)=\mathbb E_{\mathbb Q}\left[\frac{1}{\alpha}\log\left(\frac{d\mathbb Q}{d\mathbb P}\right)+\Phi_0 (M_0-S_0)+\frac{1}{2\Lambda}\int_{0}^T |M_t-S_t|^2 dt\right].
$$

Thus, in order to prove the Lemma it is sufficient to show that there exists a unique minimizer for
$\Psi:C\rightarrow \mathbb R$.
Let us notice that the convexity of the map $(z,y)\rightarrow y^2/z$
implies the convexity of $\Psi$.
From the strict convexity of the functions $z\rightarrow z\log z$ and the map $y\rightarrow y^2$ it follows that
$\Psi$ is strictly convex and so
the uniqueness of a minimizer is immediate. It remains to prove the existence of a minimizer.

Let $(Z^n,Y^n)\in C$, $n\in\mathbb N$ be a sequence such that
\begin{equation}\label{3.12}
\lim_{n\rightarrow\infty}\Psi(Z^n,Y^n)=\inf_{(Z,Y)\in C}\Psi(Z,Y).
\end{equation}
Assumption \ref{asm3.1}
and (\ref{3.12}) yield that without loss of generality we can assume that
$\sup_{n\in\mathbb N}\mathbb E_{\mathbb P}\left [ Z^n\log {Z^n}\right]<\infty$ and
$\sup_{n\in\mathbb N}\mathbb E_{\mathbb P}\left[\int_{0}^T \frac{|Y^n_t|^2}
{\mathbb E_{\mathbb P}\left[Z^n|\mathcal F_t\right]}dt\right]<\infty$.
Thus, the de la Vall\`{e}e--Poussin
criterion ensures that $Z^n$, $n\in\mathbb N$ are uniformly integrable.
Let us argue that for any $s<T$ the random variables
$Y^n_s$, $n\in\mathbb N$ are uniformly integrable

Fix $s<T$. From the Jensen inequality and fact that $Y^n$ is a martingale it follows that for any given $n$
the function $t\rightarrow\mathbb E_{\mathbb P}\left[\frac{|Y^n_t|^2}{\mathbb E_{\mathbb P}\left[Z^n|\mathcal F_t\right]}\right]$ is non decreasing, and so
$\sup_{n\in\mathbb N}\mathbb E_{\mathbb P}\left[ \frac{|Y^n_s|^2}{\mathbb E_{\mathbb P}\left[Z^n|\mathcal F_s\right]}\right]<\infty.$
This together with the inequality $\sup_{n\in\mathbb N}\mathbb E_{\mathbb P}\left [ Z^n\log {Z^n}\right]<\infty$
gives that
$\sup_{n\in\mathbb N}\mathbb E_{\mathbb P}[ g|Y^n_s|]<\infty$ where
$g(y):=\inf_{z>0} \left\{\frac{y^2}{z}+z\log z\right\}$, $y>0$.
For a given $y>0$, the function $z\rightarrow \frac{y^2}{z}+z\log z$ is convex
and attains its minimum at the unique $z=z(y)$ which satisfies
$1+\log z=\frac{y^2}{z^2}$. Obviously $\lim_{y\rightarrow\infty}z(y)=\infty$
and $y=z(y)\sqrt{1+\log(z(y))}$. Thus,
$\lim_{y\rightarrow\infty}\frac{g(y)}{y}\geq \lim_{z\rightarrow\infty} \frac{z\log z}{z\sqrt{1+\log z }}=\infty.$
Hence, from the de la Vall\`{e}e--Poussin
criterion we conclude that
$Y^n_s$, $n\in\mathbb N$ are uniformly integrable.

Next, define the sequence of random varialbes $(H_n)_{n\in\mathbb N}\in L^1(dt\otimes \mathbb P, [0,2T]\times \Omega)$ by
$$H^n(t,\omega) := \begin{cases}
  Y^n_t(\omega)  & \mbox{if} \ \ t<T\\
  Z^n(\omega) & \ \mbox{otherwise}.
\end{cases}
$$
Observe that the relations
$\sup_{n\in\mathbb N}\mathbb E_{\mathbb P}\left [ Z^n\log {Z^n}\right]<\infty$,
$\sup_{n\in\mathbb N}\mathbb E_{\mathbb P}\left[\int_{0}^T \frac{|Y^n_t|^2}
{\mathbb E_{\mathbb P}\left[Z^n|\mathcal F_t\right]}dt\right]<\infty$
and $\lim_{y\rightarrow\infty}\frac{g(y)}{y}=\infty$ yield that
$(H^n)_{n\in\mathbb N}$ is a bounded sequence in
$L^1(dt\otimes \mathbb P, [0, 2T]\times \Omega)$.
Hence, from the well-known Komlos-argument, (see Theorem 1.3 in \cite{DS:94})
there exists a sequence
$(\hat H^n)\in conv\left(H^n,H^{n+1}...\right)$, $n\in\mathbb N$ such that
$(\hat H^n)_{n\in\mathbb N}$ converge in probability ($dt\otimes\mathbb P$) to some $\hat H \in L^1(dt\otimes \mathbb P, [0, 2T]\times \Omega)$.
From the bounded convergence theorem $\arctan(\hat H^n)\rightarrow\arctan(H)$ in $L^1(dt\otimes \mathbb P, [0, 2T]\times \Omega)$.
Thus,
from the Fubini theorem we obtain that there exists a dense set $\mathcal I\subset [0,2T]$ such that
for any $t\in \mathcal I$
$(\hat H^n_t)_{n\in\mathbb N}$ converge in probability to $\hat H_t$.
Choose a countable subset $\mathcal J\subset \mathcal I\cap [0,T)$
of the form $\mathcal J=\{t_1<t_2<...\}$ such that $\lim_{n\rightarrow \infty} t_n=T$.

We conclude
that there exist convex combinations (same combinations as for $H^n$)
$(\hat Z^n,\hat Y^n)\in conv\left((Z^{n},Y^{n}),(Z^{n+1},Y^{n+1})...\right)$, $n\in\mathbb N$,
such that
the sequence $(\hat Z^n)_{n\in\mathbb N}$ converges in probability to
some $\hat Z$ and
for any $m\in\mathbb N$ the sequence $(\hat Y^n_{t_m})_{n\in\mathbb N}$
converges in probability to some $U_m$.
From the uniform integrability of the sequences $(Z^n)_{n\in\mathbb N}$ and
 $(Y^n_{t_m})_{n\in\mathbb N}$, $m\in\mathbb N$ we conclude
\begin{equation}\label{3.14}
\hat Z^n\rightarrow \hat Z \ \mbox{in} \ L^1(\mathbb P)
\end{equation}
and for any $m\in\mathbb N$
\begin{equation}\label{3.15}
\hat Y^n_{t_m}\rightarrow U_m\ \mbox{in} \ L^1(\mathbb P).
\end{equation}
Notice that (\ref{3.14}) implies $\mathbb E_{\mathbb P}[\hat Z]=1$.
 Moreover, the function $x\rightarrow x\log x$, $x>0$ is bounded from below,
 thus from the Fatou lemma and the convexity of the function $x\rightarrow x\log x$ we get
$\mathbb E_{\mathbb P}\left [ \hat Z\log {\hat Z}\right]\leq \sup_{n\in\mathbb N}
\mathbb E_{\mathbb P}\left [ Z^n\log {Z^n}\right]<\infty$.

Next, define the process $\hat Y=(\hat Y_t)_{0\leq t<T}$ by
$\hat Y_t:=\sum_{m=1}^{\infty}\mathbf 1_{\{t\in [t_{m-1},t_{m})\}} \mathbb E_{\mathbb P}[U_m|\mathcal F_t]$ where we set
$t_0:=0$.
Clearly, for any $n$ the process $\hat Y^n=\hat Y^n_{[0,T)}$ is a martingale (convex combination of martingales), and so from (\ref{3.15}) we obtain
that $\hat Y=\hat Y_{[0,T)}$ is a martingale and
\begin{equation}\label{new}
\hat Y^n_t\rightarrow \hat Y_t \ \mbox{in} \ L^1(\mathbb P) \ \forall t\in [0,T].
\end{equation}
From (\ref{3.14}) we get
\begin{equation}\label{new1}
\mathbb E_{\mathbb P}[\hat Z^n|\mathcal F_t]\rightarrow \mathbb E_{\mathbb P}[\hat Z|\mathcal F_t] \ \mbox{in}  \ L^1(\mathbb P) \ \forall t\in [0,T].
\end{equation}

By combining the Fatou lemma, the convexity of $\Psi$, (\ref{3.12})
and (\ref{new})--(\ref{new1})
we obtain
$$\Psi(\hat Z,\hat Y)\leq \lim_{n\rightarrow\infty}\Psi(\hat Z^n,\hat Y^n)=
\inf_{(Z,Y)\in C}\Psi(Z,Y).$$
A priori it might happen that
$dt\otimes \mathbb P\left(\mathbb E_{\mathbb P}[\hat Z|\mathcal F_t]=0\right)>0$ and so we need to be careful with the definition of
$\Psi(\hat Z,\hat Y)$.
From (\ref{new})--(\ref{new1}) it follows that we have the convergence in probability
$\mathbb E_{\mathbb P}[\hat Z^n|\mathcal F_{\cdot}]\rightarrow \mathbb E_{\mathbb P}[\hat Z|\mathcal F_{\cdot}]$ and
$\hat Y^n\rightarrow \hat Y$
with respect to the product measure $dt\otimes\mathbb P$. Hence, by taking a subsequence (which for simplicity we still denote by $n$)
we can assume that
$\mathbb E_{\mathbb P}[\hat Z^n|\mathcal F_{\cdot}]\rightarrow \mathbb E_{\mathbb P}[\hat Z|\mathcal F_{\cdot}]$
and $\hat Y^n\rightarrow \hat Y$
$dt\otimes\mathbb P$ a.s.
Since $\lim_{n\rightarrow\infty}\Psi(\hat Z^n,\hat Y^n)<\infty$,
from the Fatou lemma it follows that
$\mathbb E_{\mathbb P}\left[\int_{0}^T\left(\lim\inf_{n\rightarrow \infty}\frac{|\hat Y^n_t|^2}{\mathbb E_{\mathbb P}[\hat Z^n|\mathcal F_t]}\right)dt
\right]<\infty$.
In particular $\lim\inf_{n\rightarrow \infty}\frac{|\hat Y^n_t|^2}{\mathbb E_{\mathbb P}[\hat Z^n|\mathcal F_t]}<\infty$
$dt\otimes \mathbb P$ a.s. This together
with the above convergence of the sequences $(\mathbb E_{\mathbb P}[\hat Z^n|\mathcal F_{\cdot}])_{n\in\mathbb N}$ and
$(\hat Y^n)_{n\in\mathbb N}$
yields the implication
$\mathbb E_{\mathbb P}[\hat Z|\mathcal F_t]=0\Rightarrow\hat Y_t=0$
$dt\otimes\mathbb P$ a.s.
Thus, we set
$\frac{\hat Y_t}{\mathbb E_{\mathbb P}[\hat Z|\mathcal F_t]}:=0$ if
$\mathbb E_{\mathbb P}[\hat Z|\mathcal F_t]=0$.

Finally, in order to complete the proof it remains to show that $\hat Z>0$ a.s.
To this end define the function $f:[0,1]\rightarrow\mathbb R$
by
$f(\alpha):=\Psi(\alpha +(1-\alpha)\hat Z,\hat Y)$, $\alpha\in [0,1].$
From the convexity of $\Psi$ it follows that $f$ is convex. The inequality $\Psi(\hat Z,\hat Y)\leq\inf_{( Z, Y)\in C}\Psi( Z, Y)$ yields that
the right--hand derivative $f'(0+)\geq 0$. Moreover, from the monotone (derivative of a convex function) convergence theorem it follows that
we can interchange derivative and expectation. Thus,
\begin{align*}
&0\leq f'(0+)= \mathbb E_{\mathbb P}\left[\frac{1}{\alpha}(1-\hat Z)\log\hat Z-\Phi_0S_0(1-\hat Z)\right]\\
&+\frac{1}{2\Lambda}\mathbb E_{\mathbb P}\left[(1-\hat Z)\int_{0}^T
S^2_t dt+\int_{0}^T\mathbf{1}_{\{\mathbb E_{\mathbb P}[\hat Z|\mathcal F_t]>0\}}
\left(\frac{\hat Y^2_t}{\mathbb E_{\mathbb P}[\hat Z|\mathcal F_t]}-\frac{\hat Y^2_t}{\mathbb E^2_{\mathbb P}[\hat Z|\mathcal F_t]}\right)dt\right].
\end{align*}
We conclude that $\mathbb E_{\mathbb P}[\log\hat Z]>-\infty$ and complete the proof.
\end{proof}

Now, we have all the ingredients for the proof of \textbf{Theorem \ref{thm2.1}.}
\begin{proof}
Let $(\hat{\mathbb Q}\in\mathcal Q,\hat M\in\mathcal M^{\hat{\mathbb Q}}_{[0,T]})$
be the minimizer from Lemma \ref{lem3.2}. Denote
$$D:=
\mathbb E_{\hat{\mathbb Q}}\left[\frac{1}{\alpha}\log\left(\frac{d\hat{\mathbb Q}}{d\mathbb P}\right)+\Phi_0(\hat M_0-S_0)+\frac{1}{2\Lambda}\int_{0}^T
|\hat {M}_t-S_t|^2 dt\right].$$
Let us show that there exists $\hat\phi\in\mathcal A_{\Phi_0}$ such that
\begin{equation}\label{3.16}
V^{\Phi_0,\hat\phi}_T\geq D-\frac{1}{\alpha} \log\left(\frac{d\hat{\mathbb Q}}{d\mathbb P}\right).
\end{equation}
We apply Lemma \ref{lem3.1}
for $X:= D-\frac{1}{\alpha} \log\left(\frac{d\hat{\mathbb Q}}{d\mathbb P}\right).$
 Clearly, $X$ satisfies
(\ref{3.1+}), and so we need to show that for any $\mathbb Q\in\mathcal Q$ and
$M\in \mathcal M^{\mathbb Q}_{[0,T]}$ we have
\begin{equation}\label{3.17}
\mathbb E_{\mathbb Q}
\left[\frac{1}{\alpha} \log\left(\frac{d\hat{\mathbb Q}}{d\mathbb P}\right)+\Phi_0(M_0-S_0)+\frac{1}{2\Lambda}\int_{0}^T |M_t-S_t|^2 dt\right]\geq D.
\end{equation}
Choose $\mathbb Q\in\mathcal Q$ and
$M\in \mathcal M^{\mathbb Q}_{[0,T]}$.
Define
$(Z,Y),(\hat Z,\hat Y)\in C$ by
$Z=\frac{d\mathbb Q}{d\mathbb P}$,
$\hat Z=\frac{d\hat{\mathbb Q}}{d\mathbb P}$,
$Y_t=M_t\mathbb E_{\mathbb P}[Z|\mathcal F_t]$ and $\hat Y_t=\hat M_t\mathbb E_{\mathbb P}[\hat Z|\mathcal F_t]$, $t<T$.
Define the convex function $h:[0,1]\rightarrow\mathbb R_{+}$
by $h(\alpha):=\Psi\left(\alpha Z+(1-\alpha)\hat Z,\alpha Y+(1-\alpha)\hat Y\right)$, $\alpha\in [0,1].$ The function $h$ attains its minimum
at $\alpha=0$, and so
$h'(0+)\geq 0$. Again, the monotone convergence theorem allows
us to interchange derivative and expectation. Thus,
\begin{align}\label{3.18}
&0\leq h'(0+)= \mathbb E_{\mathbb P}\left[\frac{1}{\alpha}(Z-\hat Z)\log\hat Z+\Phi_0(Y_0-\hat Y_0)-\Phi_0 S_0(Z_0-\hat Z_0)\right]\nonumber\\
&+\frac{1}{2\Lambda}\mathbb E_{\mathbb P}\left[\int_{0}^T
\left((Z-\hat Z) S^2_t-2(Y_t-\hat Y_t)S_t\right)dt\right]\nonumber\\
&+\frac{1}{2\Lambda}\mathbb E_{\mathbb P}\left[\frac{1}{2\Lambda}\int_{0}^T
\frac{2 \hat Y_t(Y_t-\hat Y_t)\mathbb E_{\mathbb P}[\hat Z|\mathcal F_t]-
\mathbb E_{\mathbb P}[Z-\hat Z|\mathcal F_t]\hat Y^2_t
}{\mathbb E^2_{\mathbb P}[\hat Z|\mathcal F_t]}\right].
\end{align}
Observe that for any $t<T$,
$2Y_t\hat Y_t \mathbb E_{\mathbb P}[\hat Z|\mathcal F_t]\leq \mathbb E_{\mathbb P}[Z|\mathcal F_t]\hat Y^2_t+\frac{Y^2_t \mathbb E^2_{\mathbb P}[\hat Z|\mathcal F_t]}{\mathbb E_{\mathbb P}[Z|\mathcal F_t]}.$
This together with (\ref{3.18})
gives
\begin{align*}
&0\leq \mathbb E_{\mathbb P}\left[\frac{1}{\alpha}(Z-\hat Z)\log\hat Z+\Phi_0(Y_0-\hat Y_0)-\Phi_0 S_0(Z_0-\hat Z_0)\right]\\
&+\frac{1}{2\Lambda}\mathbb E_{\mathbb P}\left[\int_{0}^T
\left((Z-\hat Z)S^2_t-2(Y_t-\hat Y_t)S_t\right)dt\right]\\
&+\frac{1}{2\Lambda}\mathbb E_{\mathbb P}\left[\int_{0}^T
\left(\frac{Y^2_t}{\mathbb E_{\mathbb P}[Z|\mathcal F_t]}-\frac{\hat Y^2_t}{\mathbb E_{\mathbb P}[\hat Z|\mathcal F_t]}\right)dt
\right]
\end{align*}
which is exactly (\ref{3.17}). We conclude that (\ref{3.16}) holds true, and so
\begin{equation}\label{final1}
\mathbb E_{\mathbb P}\left[e^{-\alpha V^{\Phi_0,\hat\phi}_T}\right]\leq e^{-\alpha D}.
\end{equation}

We arrive at the final step of the proof. Choose $\phi\in\mathcal A_{\Phi_0}$.
Without loss of generality assume that
$\mathbb E_{\mathbb P}\left[e^{-\alpha V^{\Phi_0,\phi}_T}\right]<\infty$ and so,
similar arguments as in the proof of Lemma \ref{lem3.1} yield
$\mathbb E_{\hat{\mathbb Q}}\left[\int_{0}^T\phi^2_t dt \right]<\infty.$
Let us argue that for any $\gamma>0$
\begin{align}\label{3.gam}
&\mathbb E_{\mathbb P}\left[e^{-\alpha V^{\Phi_0,\phi}_T}\right]\geq\nonumber\\
 &\alpha\gamma\mathbb E_{\hat {\mathbb Q}}\left[\Phi_0 S_0+\int_{0}^T S_t\phi_t dt+\frac{\Lambda}{2}\int_{0}^T \phi^2_t dt\right]
-\mathbb E_{\mathbb P}\left[\gamma \hat Z\left(\log(\gamma \hat Z)-1\right)
\right]=\nonumber\\
&\alpha\gamma\mathbb E_{\hat {\mathbb Q}}\left[\Phi_0 (S_0-\hat M_0)+\int_{0}^T (S_t-\hat M_t)\phi_t  dt+\frac{\Lambda}{2}\int_{0}^T \phi^2_t dt\right]-\gamma(\log \gamma-1)-\gamma \mathbb E_{\hat {\mathbb Q}}\left[\log \hat Z\right]\nonumber\\
&\geq\alpha \gamma\mathbb E_{\hat{\mathbb Q}}\left[ \Phi_0(S_0-\hat M_0)-\frac{1}{2\Lambda}\int_{0}^T|\hat M_t-S_t|^2 dt\right]-\gamma(\log \gamma-1)-\gamma\mathbb E_{\hat{\mathbb Q}}\left[\log \hat Z\right].
\end{align}
Indeed the first inequality follows from the simple inequality $e^x\geq xy-y(\log y-1)$, $x\in\mathbb R$, $y>0$.
The equality is due to
$\mathbb E_{\hat {\hat{\mathbb Q}}}\left[\Phi_0 \hat M_0+\int_{0}^T \hat M_t \phi_t dt\right]=0$ (for this we need the bound
$\mathbb E_{\hat{\mathbb Q}}\left[\int_{0}^T\phi^2_t dt \right]<\infty$).
The last inequality
 follows from the maximization of the quadratic pattern in $\phi$.

Optimizing (\ref{3.gam}) in $\gamma>0$ we arrive at
\begin{equation}\label{final2}
\mathbb E_{\mathbb P}\left[e^{-\alpha V^{\Phi_0,\phi}_T}\right]\geq e^{-\alpha D}.
\end{equation}
Since $\phi\in\mathcal A_{\Phi_0}$ was arbitrary, from (\ref{final1}), (\ref{final2}) and the fact that
$(\hat{\mathbb Q}\in\mathcal Q,\hat M\in\mathcal M^{\hat{\mathbb Q}}_{[0,T]})$
is the minimizer from Lemma \ref{lem3.2} we obtain
(\ref{3.1}). Moreover, let us notice that
there is an equality in (\ref{3.gam}) if and only if
$\phi=\frac{\hat M-S}{\Lambda}$ $dt\otimes\mathbb P$
a.s. This yields (\ref{3.2})
and completes the proof.
\end{proof}

\section{Explicit Computations in the Bachelier Model}\label{sec:3}
In this section we assume that the
probability space $(\Omega, \mathcal{F}, \mathbb P)$
carrying a one-dimensional Wiener process
$W=(W_t)_{t \in [0,T]}$
and the
filtration
$(\mathcal{F}_t)_{t \in [0,T]}$
is the natural augmented filtration generated by $W$.
The risky
asset $S$ is given by
\begin{equation}\label{2.vas}
S_t=S_0+\sigma W_t+\mu t, \ \ t\in [0,T]
\end{equation}
where $S_0\in\mathbb R$ is the initial asset price, $\sigma>0$ is the constant volatility and $\mu\in\mathbb R$ is the constant
drift.

Consider a European contingent claim with the quadratic payoff
$\mathcal X=\kappa S^2_T$ where $\kappa\in \left(0,\frac{1}{2\alpha\sigma^2  T}\right)$ is a constant.
We say that $\hat\phi\in\mathcal A_{\Phi_0}$ is
 a utility-based optimal hedging strategy if
 \begin{equation*}
 \mathbb E_{\mathbb P}\left[e^{\alpha\left(\mathcal X-V^{\Phi_0,\hat\phi}_T\right)}\right]=\inf_{\phi\in\mathcal A_{\Phi_0}}
 \mathbb E_{\mathbb P}\left[e^{\alpha\left(\mathcal X-V^{\Phi_0,\phi}_T\right)}\right].
 \end{equation*}
 \begin{thm}\label{thm3.1}
Let $\rho:=\frac{\alpha\sigma^2}{\Lambda}$
be the risk-liquidity ratio.
The utility-based optimal hedging strategy $\hat \phi_t$, $t\in [0,T)$ is unique and given by
the feedback form
\begin{equation}\label{port}
\hat\phi_t=\frac{\left(2\kappa S_t+\frac{\mu}{\alpha\sigma^2}\right)\tanh\left(\sqrt\rho(T-t)/2\right)-\left(\coth\left(\sqrt\rho (T-t)\right)-2\Lambda\sqrt{\rho}\kappa\right)\hat\Phi_t}{\frac{1}{\sqrt\rho}-4\kappa\Lambda \tanh\left(\sqrt\rho(T-t)/2\right)}
\end{equation}
where $\hat\Phi_t:=\Phi_0+\int_{0}^t\hat\phi_s ds$, $t\in [0,T].$
\end{thm}
Our feedback description (\ref{port}) can be interpreted as follows.
From the simple inequality $\tanh (z)<z$, $\forall z>0$
and the assumption $\kappa\in \left(0,\frac{1}{2\alpha\sigma^2 T}\right)$ it follows that the
denominator in (\ref{port}) is positive and
$\coth\left(\sqrt\rho (T-t)\right)-2\Lambda\sqrt{\rho}\kappa>0$. Thus, the optimal trading strategy is a mean reverting strategy towards
the process $\frac{\left(2\kappa S_t+\frac{\mu}{\alpha\sigma^2}\right)\tanh\left(\sqrt\rho(T-t)/2\right)}
{\coth\left(\sqrt\rho (T-t)\right)-2\Lambda\sqrt{\rho}\kappa}$, $t\in [0,T]$. This
process can be viewed as a tradeoff between the optimal trading strategy
in the frictionless case
$2\kappa S_t+\frac{\mu}{\alpha\sigma^2}$, $t\in [0,T]$
and the liquidation requirement.

Next, we prove \textbf{Theorem \ref{thm3.1}.}
\begin{proof}
First, from the assumption $\kappa\in \left(0,\frac{1}{2\alpha\sigma^2 T}\right)$ it follows that
$\mathbb E_{\mathbb P}\left[e^{\alpha\mathcal X}\right]<\infty$. Thus,
define the probability measure $\tilde{\mathbb P}$ by
$\frac{d\tilde{\mathbb P}}{d\mathbb P}:=\frac{e^{\alpha \mathcal X}}{\mathbb E_{\mathbb P}\left[e^{\alpha\mathcal X}\right]}.$
Observe that
for any probability measure
$\mathbb Q\sim\mathbb P$ we have
$\mathbb E_{\mathbb Q}\left[\log\left(\frac{d\mathbb Q}{d\tilde{\mathbb P}}\right)\right]=
\mathbb E_{\mathbb Q}\left[\log\left(\frac{d\mathbb Q}{d\mathbb P}\right)-\alpha \mathcal X\right]+
\alpha \log\left(\mathbb E_{\mathbb P}\left[e^{\alpha\mathcal X}\right]\right)$.
From the H\"{o}lder's inequality and Assumption \ref{asm3.1} it follows that there exists $b>0$ such that
$\mathbb E_{\tilde{\mathbb P}}\left[\exp\left(b\sup_{0 \leq t\leq T} S^2_t\right)\right]<\infty.$
Hence,
by applying Theorem \ref{thm2.1} for the probability measure $\tilde{\mathbb P}$
we obtain
\begin{align}\label{3.21}
&\min_{\phi\in\mathcal A_{\Phi_0}}\left\{\frac{1}{\alpha}\log\mathbb{E}_{\mathbb P}\left[\exp\left(\alpha\left(\mathcal X- V^{\Phi_0,\phi}_T\right)\right) \right]\right\}=\\
&\sup_{\mathbb Q\in\mathcal Q}\sup_{M\in\mathcal M^{\mathbb Q}_{[0,T)}}\mathbb E_{\mathbb Q}\left[\mathcal X-\frac{1}{\alpha}\log\left(\frac{d\mathbb Q}{d\mathbb P}\right)-\Phi_0(M_0-S_0)-\frac{1}{2\Lambda}\int_{0}^T |M_t-S_t|^2 dt\right].\nonumber
\end{align}
Moreover, there exists
a unique maximizer
$\left(\hat{\mathbb Q}\in\mathcal Q,\hat M\in \mathcal M^{\hat {\mathbb Q}}_{[0,T)}\right)$
 for the right-hand side of (\ref{3.21})
 and the process
given by (\ref{3.2})
is the unique utility-based optimal hedging strategy.

Observe that by the Markov property of Brownian motion, in order to prove
Theorem \ref{thm3.1} it is sufficient to establish (\ref{port}) for $t=0$.
Thus, in view of (\ref{3.2}) it remains to establish that
\begin{equation}\label{end}
\frac{\hat M_0-S_0}{\Lambda}=\frac{\left(2\kappa S_0+\frac{\mu}{\alpha\sigma^2}\right)\tanh\left(\sqrt\rho T/2\right)-\left(\coth\left(\sqrt\rho T\right)-2\Lambda\sqrt{\rho}\kappa\right)\Phi_0}{\frac{1}{\sqrt\rho}-4\kappa\Lambda \tanh\left(\sqrt\rho T/2\right)}.
\end{equation}

To this end
let $\mathbb Q\in\mathcal Q$ and
$M\in\mathcal M^{\mathbb Q}_{[0,T)}$.
From the Girsanov theorem it follows that there exists a progressively measurable process
$\theta\in L^2(dt\otimes\mathbb Q)$ such that
$\mathbb E_{\mathbb Q}[\log(d\mathbb Q/d\mathbb P)]=\mathbb
  E_{\mathbb Q}[\int_0^T\theta^2_sds]/2<\infty$ and
  $W^{\mathbb Q}_t:=W_t-\int_{0}^t\theta_s ds$, $t\in [0,T]$ is a $\mathbb Q$-Brownian motion.
  By applying the martingale representation theorem there exists
  a process $\gamma=(\gamma_t)_{0\leq t< T}$ such that
  \begin{equation}\label{3.21+}
  M_t=M_0+\sigma\int_{0}^t \gamma_s dW^{\mathbb Q}_s, \ \ dt\otimes\mathbb P \ \ \mbox{a.s.}
  \end{equation}
  Moreover, by applying the martingale representation theorem for $\theta_t$, $t\in [0,T]$ we conclude
  that there exist
  a deterministic function $a_t$, $t\in [0,T]$ and a jointly measurable process $\beta_{t,s}$, $0\leq s\leq t\leq T$ such that
  $\beta_{t,s}$ is $\mathcal F_{t\wedge s}$ measurable and
  \begin{equation}\label{3.22}
  \theta_t=a_t+\int_{0}^t \beta_{t,s}dW^{\mathbb Q}_s, \ \ dt\otimes\mathbb P \ \ \mbox{a.s.}
  \end{equation}
  Set
  $\nu_t:=\mu t+\sigma\int_{0}^t a_s ds$, $t\in [0,T]$ and $l_{t,s}:=\int_{s}^t \beta_{u,s}du$, $0\leq s\leq t\leq T$.
  From the Fubini theorem, (\ref{2.vas}) and (\ref{3.22})
  \begin{equation}\label{3.23}
  S_t=S_0+\nu_t+\sigma \int_{0}^t \left(1+l_{t,s}\right)dW^{\mathbb Q}_s, \ \ t\in [0,T].
  \end{equation}
  Given the probability measure $\mathbb Q$ we are looking for a martingale $\tilde M\in\mathcal M^{\mathbb Q}_{[0,T)}$
  which maximizes the right-hand side of (\ref{3.21}).
  By combining (\ref{3.21+}), (\ref{3.23}) and applying the It\^{o} Isometry and the Fubini theorem we obtain
\begin{align*}
&\mathbb E_{{\mathbb Q}}\left[\Phi_0( M_0-S_0)+\frac{1}{2\Lambda}\int_{0}^T
|{M}_t-S_t|^2 dt\right]\\
&=\Phi_0(M_0-S_0)+\frac{1}{2\Lambda}\int_{0}^T \left(M_0-S_0-\nu_t \right)^2 dt\\
&+\frac{\sigma^2}{2\Lambda}\mathbb E_{\mathbb Q}\left[\int_{0}^T \int_{s}^T \left(\gamma_s-1- l_{t,s}\right)^2 dt ds\right].
\end{align*}
Given $a$ and $\beta$, we are looking for
$\hat M_0$ and $\hat\gamma$ which minimize the above right-hand side. Observe that the right-hand side is
a quadratic function   in $M_0$ and $\gamma_s$, $s\in [0,T)$. Hence, we obtain that the minimizer is unique and given by
\begin{equation}\label{3.24}
\hat M_0=S_0+\frac{1}{T}\int_{0}^T \nu_t dt-\frac{\Phi_0\Lambda}{T},
\end{equation}
and
\begin{equation}\label{3.25}
\hat\gamma_s=1+\frac{1}{T-s}\int_{s}^T l_{t,s} dt, \ \ s<T.
\end{equation}
Finally, we compute the optimal $\nu$. From the It\^{o} Isometry, the Fubini theorem and
(\ref{3.22})--(\ref{3.23})
we have
$$\mathbb E_{\mathbb Q}[S^2_T]=\left(S_0+\nu_T\right)^2
+\sigma^2 \mathbb E_{\mathbb Q}\left[\int_{0}^T\int_{s}^T \left(1+l_{t,s}\right)^2  dt ds\right]$$
and
$$\mathbb E_{\mathbb Q}\left[\log\left(\frac{d\mathbb Q}{d\mathbb P}\right)\right]=\frac{1}{2}\mathbb
  E_{\mathbb Q}\left[\int_0^T\theta^2_sds\right]=\frac{1}{2}\left(\int_{0}^Ta^2_t dt+\mathbb E_{\mathbb Q}\left[\int_{0}^{T}\int_{s}^T \beta^2_{t,s} dt ds\right]\right).
$$
These equalities together with (\ref{3.24})--(\ref{3.25}) give
\begin{align}\label{3.26}
&\sup_{M\in\mathcal M^{\mathbb Q}_{[0,T)}}\mathbb E_{\mathbb Q}\left[\mathcal X-\frac{1}{\alpha}\log\left(\frac{d\mathbb Q}{d\mathbb P}\right)-
\Phi_0( M_0-S_0)-\frac{1}{2\Lambda}\int_{0}^T | M_t-S_t|^2 dt\right]\nonumber\\
&=\mathbb E_{\mathbb Q}\left[\mathcal X-\frac{1}{\alpha}\log\left(\frac{d\mathbb Q}{d\mathbb P}\right)-
\Phi_0( \hat M_0-S_0)-\frac{1}{2\Lambda}\int_{0}^T |\hat M_t-S_t|^2 dt\right]\nonumber\\
&=I+\mathbb E_{\mathbb Q} \left[\int_{0}^T J_s ds\right]
\end{align}
where
$$
I=\kappa\left(S_0+\nu_T\right)^2-\frac{1}{2\alpha}\int_{0}^Ta^2_t dt
+\frac{1}{2\Lambda}\left(\frac{1}{T}\left(\Phi_0\Lambda-\int_{0}^T\nu_t dt\right)^2-
\int_{0}^T \nu^2_t dt\right)
$$
and
$$J_s=\kappa\sigma^2\int_{s}^T \left(1+l_{t,s}\right)^2 dt-\frac{1}{2\alpha}\int_{s}^T\beta^2_{t,s} dt+\frac{1}{2\Lambda}
\left(\frac{1}{T-s}\left(\int_{s}^T l_{t,s} dt\right)^2-
\int_{s}^T l^2_{t,s} dt\right).
$$
 From Proposition \ref{lem4.0} we conclude that the optimal $\nu$ satisfies
 (\ref{4.-2}). Hence,
from (\ref{3.24}) we obtain
(\ref{end}) and complete the proof.
\end{proof}
\begin{remark}\label{rem3.1}
By applying (\ref{3.26})
we can also compute the
the right-hand side of (\ref{3.21}). This requires computing the maximal $I$ and for any
$s\in [0,T]$ computing the maximal $J_s$.
Observe that the latter is a deterministic variational problem where the control is
$l_{\cdot,s}$, $\cdot\in [s,T]$. Computing both $I$ and $J_s$, $s\in [0,T]$
can be done by computing the value which corresponds to the optimization problem given by
(\ref{4.-1}) (for $J_s$ replace $T$ with $T-s$, $S_0,\sigma$ with $1$ and $\mu,\Phi_0$ with $0$). The computations are quite cumbersome,
hence omitted.
\end{remark}
\begin{remark}
Let us notice that
the quadratic structure of the payoff $\mathcal
X$ which is used in
(\ref{3.26}) is essential in reducing the dual
problem to a deterministic control problem.
This is due to the It\^{o} Isometry.
Although, for a general payoff the dual representation does not allow us to obtain an explicit solution,
it can still be applied for utility-based hedging problems.
For instance, the recent paper \cite{DD:22} applies Theorem \ref{thm2.1}
and computes for a general European contingent claim in the Bachelier model,
the scaling limit of the corresponding utility indifference prices for a vanishing
price impact which is inversely proportional to the risk aversion.
\end{remark}

\section{Auxiliary Result }\label{sec:4}
The following result deals with a purely deterministic setup.
\begin{prop}\label{lem4.0}
Let $\Gamma$ be the space of all continuous functions $\delta:[0,T]\rightarrow \mathbb R$
which are differentiable almost everywhere (with respect to the Lebesgue measure)
and satisfy
$\delta(0)=0$. Then, the maximizer $\hat\delta\in\Gamma$ of the optimization problem
\begin{equation}\label{4.-1}
\max_{\delta\in\Gamma}\left\{\kappa\left(S_0+\delta_T\right)^2-\frac{1}{2\alpha\sigma^2}\int_{0}^T(\dot{\delta}_t-\mu)^2 dt
+\frac{1}{2\Lambda}\left(\frac{1}{T}\left(\Phi_0\Lambda-\int_{0}^T\delta_t dt\right)^2-
\int_{0}^T \delta^2_t dt\right)\right\}
\end{equation}
is unique and satisfies
\begin{equation}\label{4.-2}
\frac{1}{T}\int_{0}^T\hat\delta_t dt=\frac{\left(2\kappa S_0+\frac{\mu}{\alpha\sigma^2}\right)\tanh(\sqrt{\rho} T/2)-
\left(\coth(\sqrt\rho T)-2\Lambda\sqrt\rho\kappa\right)\Phi_0}{\frac{1}{\sqrt{\rho}\Lambda}-4 \kappa \tanh(\sqrt{\rho} T/2)}+\frac{\Phi_0\Lambda}{T}.
\end{equation}
\end{prop}
\begin{proof}
The proof will be done in two steps.
First we will solve the optimization problem (\ref{4.-1}) for the case where $\delta_T$ and
$\int_{0}^T\delta_t dt$ are given. Then, we will find the optimal $\delta_T$ and
$\int_{0}^T\delta_t dt$.

Thus, for any $x,y$ let $\Gamma_{x,y}\subset \Gamma$ be the set of all functions
$\delta\in\Gamma$ which satisfy $\delta_T=x$ and $\int_{0}^T\delta_t dt=y$.
Consider the minimization problem
$\min_{\delta\in \Gamma_{x,y}} \int_{0}^T H(\dot \delta_t,\delta_t)dt$
where $H(u,v):=\frac{1}{2\alpha\sigma^2}(u-\mu)^2+\frac{1}{2\Lambda} v^2$ for $u,v\in\mathbb R$.
This optimization problem  is convex and so it has a unique solution which has
to satisfy the Euler–Lagrange equation (for details see \cite{GF:63})
$\frac{d}{dt}\frac{\partial H}{\partial \dot \delta_t}=\lambda+\frac{d}{dt}\frac{\partial H}{\partial \delta_t}$
for some constant $\lambda>0$ (lagrange multiplier due to the constraint $\int_{0}^T \delta_t dt=y$).
Thus, the optimizer solves the ODE
$\ddot{\delta}_t-\rho \delta\equiv const$ (recall the risk-liquidity ration $\rho=\frac{\alpha\sigma^2}{\Lambda}$).
From the standard theory it follows that
\begin{equation}\label{sol}
\delta_t=c_1\sinh(\sqrt\rho t)+c_2\sinh(\sqrt\rho(T-t))+c_3
\end{equation}
for some constants $c_1,c_2,c_3$.
From the three constraints
$\delta_0=0$, $\delta_T=x$ and $\int_{0}^T \delta_t dt=y$ we obtain
\begin{equation}\label{sol1}
c_1=\frac{x-c_3}{\sinh(\sqrt\rho T)}, \ \ c_2=-\frac{c_3}{\sinh(\sqrt\rho T)} \ \
\mbox{and} \  \ c_3=\frac{\sqrt\rho y-x\tanh(\sqrt\rho T/2)}{\sqrt\rho T-2\tanh(\sqrt\rho T/2)}.
\end{equation}

We argue that
\begin{align}\label{4.1}
&\rho\int_{0}^T \delta^2_t dt+\int_{0}^T \dot{\delta}^2_tdt=\rho\int_{0}^T \left((\delta_t-c_3)+c_3\right)^2 dt+\int_{0}^T \dot{\delta}^2_tdt\nonumber\\
&=\frac{\sqrt\rho}{2}\left(c^2_1+c^2_2\right)\sinh\left(2\sqrt\rho T\right)-2c_1c_2\sqrt\rho\sinh(\sqrt\rho T)-\rho c^2_3 T+2\rho c_3 y\nonumber\\
&=\sqrt{\rho}x^2\coth(\sqrt\rho T)+2\sqrt\rho c_1 c_2\sinh(\sqrt\rho T)\left(\cosh(\sqrt\rho T)-1\right)-\rho c^2_3 T+2\rho c_3 y\nonumber\\
&=\sqrt{\rho}x^2\coth(\sqrt\rho T)+\left(2\sqrt\rho \tanh(\sqrt \rho T/2)-\rho T\right)c^2_3\nonumber\\
&+2\left(\rho y-\sqrt\rho \tanh(\sqrt \rho T/2)x\right)c_3\nonumber\\
&=\sqrt{\rho}\left(x^2\coth(\sqrt\rho T)+\frac{\left(x\tanh(\sqrt\rho T/2)-\sqrt\rho y\right)^2}{\sqrt\rho T-2\tanh(\sqrt\rho T/2)}\right).
\end{align}
Indeed, the first equality is obvious. The second equality follows from (\ref{sol}) and simple computations. The third equality is due
to $c_1-c_2=\frac{x}{\sinh(\sqrt\rho T)}$. The fourth equality is due to
$c_1c_2=\frac{c^2_3-xc_3}{\sinh^2(\sqrt\rho T)}$. The last equality follows from substituting $c_3$.

By combining (\ref{4.1}) and the simple equality
$$\frac{1}{2\alpha\sigma^2}\int_{0}^T (\dot \delta_t-\mu)^2 dt +\frac{1}{2\Lambda}
\int_{0}^T \dot \delta_t^2 dt=\frac{\mu^2 T}{2\alpha\sigma^2}-\frac{\mu x}{\alpha\sigma^2}+\frac{1}{2\alpha\sigma^2}\left(\rho\int_{0}^T \delta^2_t dt+\int_{0}^T \dot{\delta}^2_tdt\right)
$$
we obtain
\begin{align}\label{4.1+}
&\min_{\delta\in \Gamma_{x,y}} \int_{0}^T H(\dot \delta_t,\delta_t)dt=\nonumber\\
&\frac{\mu^2 T}{2\alpha\sigma^2}-\frac{\mu x}{\alpha\sigma^2}+\frac{1}{2\Lambda\sqrt\rho}\left(x^2\coth(\sqrt\rho T)+\frac{\left(x\tanh(\sqrt\rho T/2)-\sqrt\rho y\right)^2}{\sqrt\rho T-2\tanh(\sqrt\rho T/2)}\right).
\end{align}

We arrive at the final step of the proof. In view of (\ref{4.1+}), the optimization problem (\ref{4.-1})
is reduced to finding
$x:=\delta_T$ and $y:=\int_{0}^T\delta_t dt$ which maximize the quadratic form
$$A x^2+B y^2+2C xy + \eta x+\theta y$$
where
\begin{align*}
&A:=\kappa-\frac{1}{2\Lambda\sqrt\rho}\left(\coth(\sqrt\rho T)+\frac{\tanh^2(\sqrt\rho T/2)}{\sqrt\rho T-2 \tanh(\sqrt\rho T/2)}\right)\\
&=-\frac{\sqrt\rho T\coth(\sqrt\rho T)+4\Lambda\sqrt\rho\kappa \tanh(\sqrt\rho T/2)-1-2\Lambda\rho T\kappa}{{2\Lambda\sqrt\rho\left(\sqrt\rho T-2 \tanh(\sqrt\rho T/2)\right)}},
\\
&B:=\frac{1}{2\Lambda T}-\frac{\sqrt\rho}{2\Lambda\left(\sqrt\rho T-2 \tanh(\sqrt\rho T/2)\right)}=
-\frac{\tanh(\sqrt\rho T/2)}{\Lambda T \left(\sqrt\rho T-2 \tanh(\sqrt\rho T/2)\right)},\\
&C:=\frac{\tanh(\sqrt\rho T/2)}{2\Lambda\left(\sqrt\rho T-2 \tanh(\sqrt\rho T/2)\right)}, \ \ \eta:=2\kappa S_0+\frac{\mu}{\alpha\sigma^2} \ \ \mbox{and} \ \ \theta:=-\frac{\Phi_0}{T}.
\end{align*}
Simple computations give
\begin{align*}
&AB-C^2=-\frac{\kappa\tanh(\sqrt\rho T/2)}{\Lambda T\left(\sqrt\rho T-2 \tanh(\sqrt\rho T/2)\right)}\\
&+\frac{\coth(\sqrt\rho T)\tanh(\sqrt\rho T/2)}{2\sqrt\rho \Lambda^2 T \left(\sqrt\rho T-2 \tanh(\sqrt\rho T/2)\right)}-
\frac{\tanh^2(\sqrt\rho T/2)}{4\sqrt\rho \Lambda^2 T \left(\sqrt\rho T-2 \tanh(\sqrt\rho T/2)\right)}\\
&=\frac{\frac{1}{4\sqrt\rho\Lambda}-\kappa\tanh(\sqrt\rho T/2)}{\Lambda T\left(\sqrt\rho T-2 \tanh(\sqrt\rho T/2)\right)}.
\end{align*}
From the inequality $z>\tanh(z)z$, $z>0$ and the assumption
$\kappa\in \left(0,\frac{1}{2\alpha\sigma^2 T}\right)$ we obtain that $B<0$ and $AB-C^2>0$. Thus, the above quadratic form has a unique maximizer
$(\bar x,\bar y):=\frac{1}{2(AB-C^2)}\left(C\theta-B\eta,C\eta-A\theta\right).$
We conclude that the optimization problem
(\ref{4.-1}) has a unique solution which is given by (\ref{sol})--(\ref{sol1}) for
$(x,y):=(\bar x,\bar y)$.
Moreover, direct computations yield
$$\frac{\bar y}{T}=\frac{\left(2\kappa S_0+\frac{\mu}{\alpha\sigma^2}\right)\tanh(\sqrt{\rho} T/2)-
\left(\coth(\sqrt\rho T)-2\Lambda\sqrt\rho\kappa\right)\Phi_0}{\frac{1}{\sqrt{\rho}\Lambda}-4 \kappa \tanh(\sqrt{\rho} T/2)}+\frac{\Phi_0\Lambda}{T}$$
and (\ref{4.-2}) follows.
\end{proof}

\end{document}